\def\year{2019}\relax
\documentclass[letterpaper,anonymous]{article} %DO NOT CHANGE THIS
\usepackage{aaai19}  %Required
\usepackage{times}  %Required
\usepackage{helvet}  %Required
\usepackage{courier}  %Required
\usepackage{url}  %Required
\usepackage{graphicx}  %Required
\usepackage{textcomp}
\usepackage{verbatim}
\usepackage{mathtools}
\usepackage{amssymb}
\usepackage{amsmath}

\usepackage{algorithmic}        % use algorithmic
\usepackage[vlined,ruled]{algorithm2e}    % use algorithm2e
\SetKwProg{Fn}{Function}{}{}

\usepackage{xcolor}
\usepackage[caption=false,font=footnotesize,labelfont=sf,textfont=sf]{subfig}
\usepackage{color-edits}
\addauthor{pg}{red}
\addauthor{ac}{green}
\addauthor{dh}{orange}
\usepackage{multirow}
\usepackage{dblfloatfix}
\usepackage[export]{adjustbox}
\renewcommand{\footnotesize}{\tiny}
\usepackage{mathtools}
\usepackage{etoolbox}
\preto\align{\par\nobreak\small\noindent}
\expandafter\preto\csname align*\endcsname{\par\nobreak\small\noindent}

\usepackage{amsthm}

\theoremstyle{plain}
\newtheorem{thm}{Theorem}
\theoremstyle{plain}
\newtheorem{cor}{Corollary}
\theoremstyle{definition}
\newtheorem{defn}{Definition}

\begin{document}

\title{Graph Representation Ensemble Learning}
% \author{Anonymous \\ XYZ University\\
% \And Anonymous \\ XYZ University\\
% \And Anonymous \\ XYZ University}
\author{Palash Goyal${^*}$ \\ University of Southern California \\ palashgo@usc.edu\\
\And Di Huang${^*}$ \\ University of Southern California \\ dh\_599@usc.edu\\
\And Sujit Rokka Chhetri${^*}$\\ University of California-Irvine\\ schhetri@uci.edu\\
\AND Arquimedes Canedo\thanks{These authors contributed equally to this work.} \\ Siemens Corporate Technology \\ arquimedes.canedo@siemens.com\\
\And Jaya Shree \\ University of Southern California \\ shree@usc.edu\\
\And Evan Patterson \\ Stanford University \\ epatters@stanford.edu}
% \And Jade Master \\ University of California-Riverside\\ jadem@math.ucr.edu}

%  \author{Palash Goyal,\textsuperscript{1}
% Di Huang,\textsuperscript{1}
% Sujit Rokka Chhetri,\textsuperscript{2}
% Arquimedes Canedo,\textsuperscript{3}
% Jaya Shree,\textsuperscript{2}
% Evan Patterson,\textsuperscript{4}
% Jade Master\textsuperscript{5}\\
% \textsuperscript{1}{University of Southern California}\\
% \textsuperscript{2}{University of California-Irvine}\\
% \textsuperscript{3}{Siemens Corporate Technology}\\
% \textsuperscript{4}{Stanford University}\\
% \textsuperscript{5}{University of California-Riverside}\\
% \{palashgo, dh\_599, shree\}@usc.edu,
% schhetri@uci.edu,
% arquimedes.canedo@siemens.com,
% epatters@stanford.edu,
% jadem@math.ucr.edu}

% \title{dyngraph2vec: Capturing Network Dynamics using Dynamic Graph Representation Learning}
% \author{Anonymous Author \\ Institute \\ email@domain\\
% \And Anonymous Author \\ Institute \\ email@domain\\
% \And Anonymous Author \\ Institute \\ email@domain}

\maketitle
Representation learning on graphs has been gaining attention due to its wide applicability in predicting missing links, and classifying and recommending nodes. Most embedding methods aim to preserve certain properties of the original graph in the low dimensional space. However, real world graphs have a combination of several properties which are difficult to characterize and capture by a single approach. In this work, we introduce the problem of graph representation ensemble learning and provide a first of its kind framework to aggregate multiple graph embedding methods efficiently. We provide analysis of our framework and analyze -- theoretically and empirically -- the dependence between state-of-the-art embedding methods. We test our models on the node classification task on four real world graphs and show that proposed ensemble approaches can outperform the state-of-the-art methods by up to 8\% on macro-F1. We further show that the approach is even more beneficial for underrepresented classes providing an improvement of up to 12\%.   
\section{Introduction}\label{sec:introduction}
Graphs are used to represent data in various scientific fields including social sciences, biology and physics~\cite{Gehrke2003,freeman2000visualizing,theocharidis2009network,goyal2018recommending}.
Such representation allows researchers to gain insights about their problem. The most common tasks on graphs are link prediction, node classification and visualization.
For example, link prediction in the social domain is used to determine friendships between people. Node classification in the biology domain is used to identify genes of proteins. Similarly, visualization is used to identify communities and structure of a graph.
Recently, significant amount of work has been devoted to learning low dimensional representation of nodes in the graphs to allow the use of machine learning techniques to perform the tasks on graphs. Graph representation learning techniques embed each node in the network in a low dimensional space, and map link prediction and node classification in the network space to a nearest neighbor search and vector classification in the embedding space~\cite{goyal2017graph}. Several of these techniques have showed state-of-the-art performance on graph tasks~\cite{Grover2016,Ou2016}.

\begin{figure}
\centering
  \includegraphics[width=0.4\textwidth]{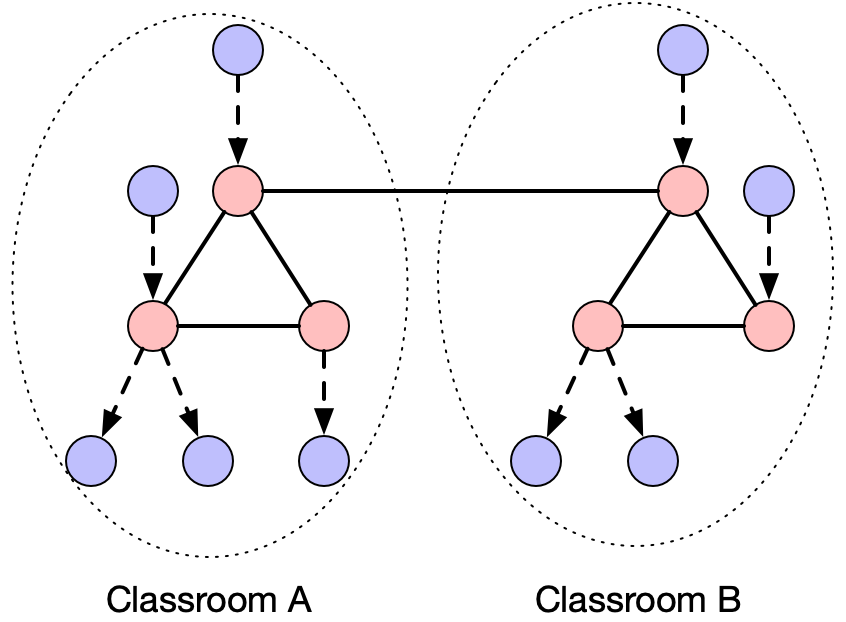}
  \vspace{-1em}
  \caption{Example motivating the need of ensemble learning. The graph represents interactions in classrooms (in red) and family trees of students (in purple). Such complex combination of community and structure requires combination of multiple embedding methods for accuracy.}
  \label{fig:intro_example}
  \vspace{-.6cm}
\end{figure}

State-of-the-art techniques in graph representation learning define some characteristics of the graphs they aim to capture and define an objective function to learn these features in the low-dimensional embedding.
For example, HOPE~\cite{Ou2016} preserves higher order proximity between nodes using the singular value decomposition of the similarity matrix.
Similarly, \textit{node2vec}~\cite{Grover2016} captures the similarity of nodes using random walks on the graph.
However, real world graphs do not follow a simple structure and can be layered with several categories of properties with complex interactions between them.
It has been shown that no single method outperforms other methods on all network tasks and data sets~\cite{goyal2017graph}.
We further illustrate this by the example in Figure~\ref{fig:intro_example} with a social network from two classrooms (represented by the pink color).
We also show the family links of individual students in the classroom and represent family members outside the classroom (represented by the the blue color).
Here, we consider the task of multi-label node classification with the classes classroom and role in family.
This network is complex and has both community and structural properties.
Methods such as HOPE~\cite{Ou2016} which preserve community can effectively classify the nodes into classrooms but perform poorly on family links which follow structure.
On the other hand, structure preserving methods can classify the role of an individual student in the family but puts nodes in the same classroom into separate categories.

In this work, we introduce graph representation ensemble learning.
Given a graph and a list of methods capturing various properties of the graph, we aim to learn a representation of nodes which can combine embeddings from each method such that it outperforms each of the constituent method in terms of prediction performance.
Ensemble methods have been very successful in the field of machine learning.
Methods such as AdaBoost~\cite{ratsch2001soft} and Random Forest~\cite{liaw2002classification} have shown to be much more accurate than the individual classifiers that compose them.
It has been shown that combining even the simplest but diverse classifiers can yield high performance.
However, to the best of our knowledge, no work has focused on ensemble learning on graph representation learning.

Here, we formally introduce ensemble learning on graph representation methods and provide a framework for it.
We first provide a motivation example to show that a single embedding approach is not enough for accurate predictions on a graph task and combining methods can yield improvement in performance.
We then formalize the problem and define a method to measure correlations of embeddings obtained from various approaches.
Then, we provide an upper bound on the correlation assuming certain properties of the graph.
The upper bound is used to establish the utility of our framework.
% We propose two approaches for this: (i) baseline approach, and (ii) coarsening approach.
% The baseline approach is a simple concatenation of embeddings obtained by methods which preserve diverse set of properties.
% For this, we propose the theorem that the necessary and sufficient condition for an ensemble of graph classifiers to be more accurate than the constituent members is if the individual classifiers are accurate and diverse.
% Further, we propose the coarsening approach which coarsens the graphs at multiple levels and combines embeddings from different level. This is done to increase the diversity in classifiers and capture higher order and structural features with ease.
% To ensure scalability of ensemble in graphs, we propose two heuristics: (i) greedy concatenation, and (ii) function approximated hyperparameter search.
% The greedy concatenation reduces time complexity of ensemble from $O(1+d)^n$ to $O(nd)$ where $n$ is the number of embedding methods and $d$ is the cardinality of the set of dimensions to be used rendering the ensemble linear.
% Function approximated hyperparameter search approximates the hyperparameters for subsequent layers of coarsening reducing the overhead.
We focus our experiments on the task of node classification.
We compare our method with the state-of-the-art embedding methods and show its performance on 4 real world networks including collaboration networks, social networks and biology networks. 
Our experiments show that the proposed ensemble approaches outperform the state-of-the-art methods by 8\% on macro-F1.
We further show that the approach is even more beneficial for underrepresented classes and get an improvement of 12\%.

Overall, our paper makes the following contributions:
\begin{enumerate}
	\item We introduce ensemble learning in the field of graph representation learning.
    \item We propose a framework for ensemble learning given a variety of graph embedding methods.
    \item We provide a theoretical analysis of the proposed framework and show its utility theoretically and empirically.
    \item We demonstrate that combining multiple diverse methods through ensemble achieves state-of-the-art accuracy.
    \item We publish a library, GraphEnsembleLearning~\footnote{https://github.com/dihuang0220/GraphEnsembleLearning}, implementing the framework for graph ensemble learning.
    
\end{enumerate}

% The rest of the paper is organized as follows. 
% Section \ref{sec:related} provides a summary of the methods proposed in this domain and differences with our model. 
% In Section \ref{sec:problem}, we provide the definitions required to understand the problem and models discussed next. 
% Section \ref{sec:method} introduces our  \textit{dyngraph2vec} model. The experimental setup and obtained results are described in Section \ref{sec:res}. 
% Finally, we draw our conclusions and discuss potential applications and future research directions in Section \ref{sec:discussion}.

\section{Related Work}\label{sec:related}
Methods for graph representation learning (aka graph embedding) typically vary in properties preserved by the approach and the objective function used to capture these properties.
Based on the properties, embedding methods can be divided into two broad categories: (i) community preserving, and (ii) structure preserving.
Community preserving approaches aim to capture the distances in the original graph in the embedding space.
Within this category, methods vary on the level of distances captured.
For example, Graph Factorization~\cite{Ahmed2013} and Laplacian Eigenmaps~\cite{belkin2001laplacian} preserve shorter distances (i.e., low order proximity) in the graph, whereas more recent methods such as Higher Order Proximity Embedding (HOPE)~\cite{Ou2016} and GraRep~\cite{cao2016deep} capture longer distances (i.e., high order proximity).
Structure preserving methods aim to understand the structural similarity between nodes and capture role of each node.
\emph{node2vec}~\cite{Grover2016} uses a mixture of breadth first and depth first search for this.
Deep learning methods such as Structural Deep Network Embedding (SDNE)~\cite{Wang2016} and Deep Network Graph Representation (DNGR)~\cite{cao2016deep} use deep autoencoders to preserve distance and structure.

Based on the objective function, embedding methods can be broadly divided into two categories: (i) matrix factorization, and (ii) deep learning methods.
Matrix factorization techniques represent graph as a similarity matrix and decompose it to get the embedding.
Graph Factorization and HOPE use adjacency matrix and higher order proximity matrix for this.
Deep learning methods, on the other hand, use multiple non-linear layers to capture the underlying manifold of the interactions between nodes.
SDNE, DNGR and VGAE~\cite{kipf2016variational} are examples of these methods.
Some other recent approaches use graph convolutional networks to learn graph structure~\cite{kipf2016semi,bruna2013spectral,henaff2015deep}.

In machine learning, ensemble approaches~\cite{zhou2012ensemble} are algorithms which combine the outputs of a set of classifiers.
It has been shown that ensemble of classifiers are more accurate than any of its individual members if the classifiers are accurate and diverse~\cite{hansen1990neural}.
There are several ways individual classifiers can be combined.
Broadly, they can be divided into four categories: (i) Bayesian voting, (ii) random selection of training examples, (iii) random selection of input features, and (iv) random selection of output labels.
Bayesian voting methods combine the predictions from the classifiers weighted by their confidence.
On the other hand, methods such as Random Forest~\cite{liaw2002classification} and Adaboost~\cite{ratsch2001soft} divide the training data into multiple subsets, train classifiers on each individual subset, and combine the output.
The third category of approaches divide the input set of features available to the learning algorithm~\cite{opitz1999feature}.
Finally, for data with a large number of output labels, some methods divide the set of output labels and learn individual classifiers to learn their corresponding label subset~\cite{ricci1997extending}.

In this work, we extend the concept of ensemble learning to graph representation learning and get insights into the correlations between various graph embedding methods.
Based on this, we propose ensemble methods for them and show the improvement in performance on node classification task.

\section{Motivating Example}\label{sec:motivation}
This section presents a motivational case study to highlight the effectiveness of the proposed graph representation ensemble learning on a synthetic dataset. We present the analysis by utilizing four synthetic graphs: (a) Barabasi-Albert, (b) Random Geometry (c) Stochastic Block Model, and (d) Watts Strogatz graph (see Figure \ref{fig:motiv1}). Each of these graphs exhibits a specific structural property. We use a spring layout to further elucidate the difference in the structural properties of the four different synthetic graphs. The Barabasi-Albert graph makes new connections through preferential attachment using the degree of the existing nodes. Watts Strogatz graph generates a ring of $n$ graphs with the addition of edges of each nodes with its $k$ neighbors. Stochastic Block Model creates community clusters by preserving the community structure. The Random Geometry graph generates $n$ nodes and add $m$ edges by utilizing the spatial proximity among the nodes as a measure. 
\begin{figure}[!ht]
\centering
\vspace{-1em}
\includegraphics[width=0.48\textwidth]{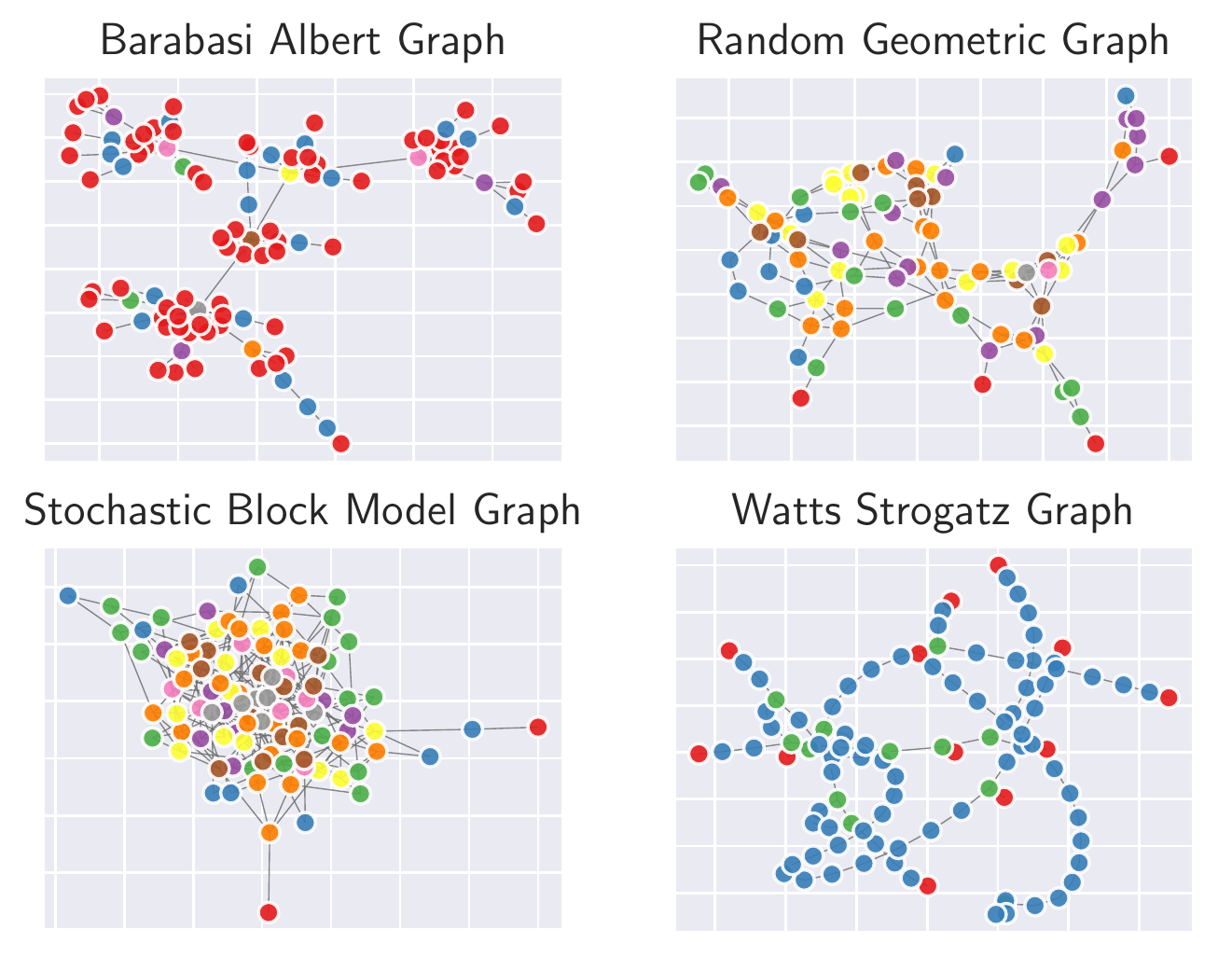}
\vspace{-1.5em}
\caption{Four synthetic graph with different graph properties (with node color representing different degrees) initially drawn using the spring layout.}
\vspace{-1em}
\label{fig:motiv1}
\end{figure}

We have generated each of the synthetic graphs with 100 nodes each. As mentioned earlier, different embedding algorithms such as Graph Factorization, Laplacian Eigenmaps, High Order Proximity Preserving, Structural Deep Network Embedding and Node2vec capture various characteristics of the graphs. Hence, a single embedding algorithm may not be able to capture the entire complex interaction. To test this hypothesis we have created two node labels for the synthetic graph. The first label is based on the degree of the graph, whereas the second label is based on the closeness centrality measure \cite{freeman1978centrality} of the graph. The centrality values are binned and the respective bins are used as node labels.

\begin{figure}[h!]%
    \centering
    \vspace{-1em}
    \subfloat[Original layout showing added edges]{{\includegraphics[width=0.24\textwidth]{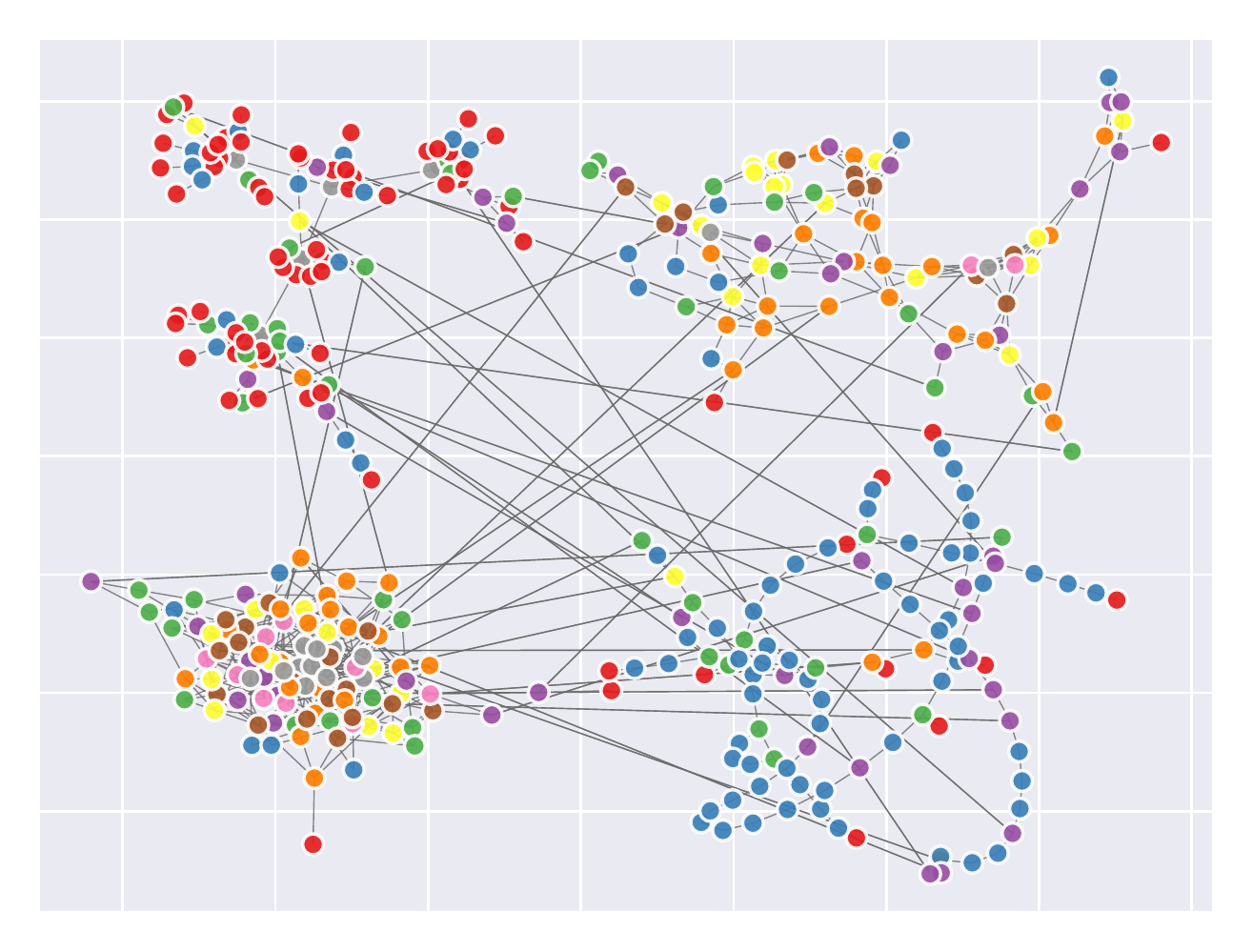}}}%
    % \qquad
    \subfloat[New spring layout]{{\includegraphics[width=0.24\textwidth]{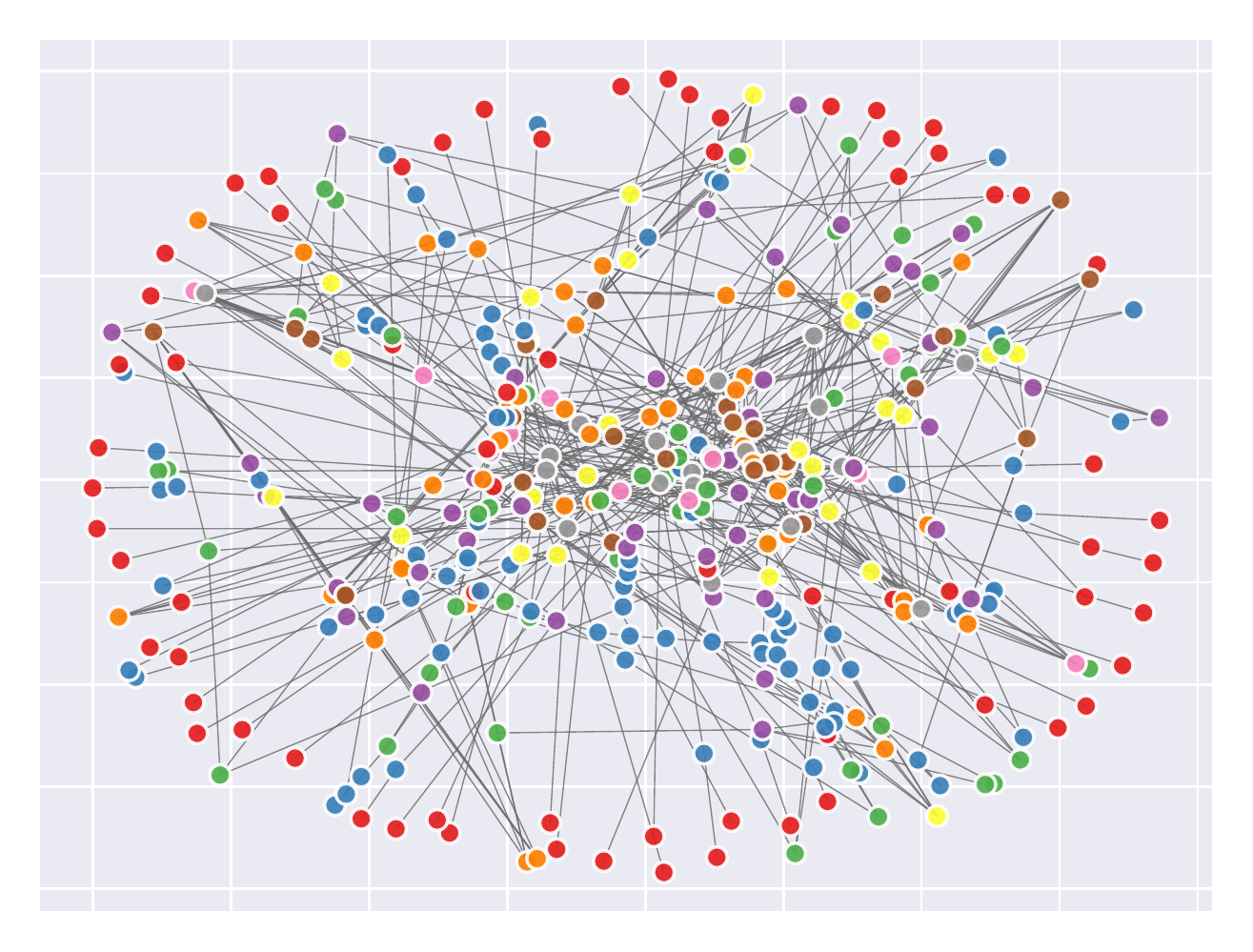}}}%
    \caption{Synthetic graphs merged together (different colors represents the updated node degree).}%
    \vspace{-1em}
    \label{fig:motiv2}%
\end{figure}

To simulate the interaction between different synthetic graphs, we have randomly selected node pairs (equal to 40\% of the total number of nodes) and added edges between them (with a probability threshold of 0.3). The addition of the edges are shown in Figure \ref{fig:motiv2}.
% \begin{figure}[!ht]
% \centering
% \vspace{-1em}
% \includegraphics[width=0.45\textwidth]{img/ensemble_merged.pdf}
% \vspace{-1em}
% \caption{Synthetic graphs merged together with addition of random edges (different colors represents the updated node degree).}
% \vspace{-1em}
% \label{fig:architecture}
% \end{figure}

% % \begin{figure}[h!]%
% %     \centering
% %     \subfloat[Synthetic graph with original layout showing added edges]{{\includegraphics[width=0.35\textwidth]{img/ensemble_merged.pdf}}}%
% %     \qquad
% %     \subfloat[Synthetic graph with new spring layout]{{\includegraphics[width=0.35\textwidth]{img/ensemble_merged_spring.pdf}}}%
% %     \caption{Synthetic graphs merged together (different colors represents the updated node degree).}%
% %     \label{fig:example}%
% % \end{figure}

\begin{table}[!ht]
\centering
% \vspace{-0.5em}

\begin{tabular}{c|c|c}
\hline
Methods & Dimensions  &  Macro-F1\\
\hline
gf & 128  & 0.127\\
lap & 32  & 0.055\\
hope & 128  & 0.157\\
sdne & 64  & \underline{0.177} \\
node2vec & 128 & 0.128\\ 
\hline
\multirow{2}{7 em}{\centering{sdne, node2vec hope,gf,lap}} & \multirow{2}{7 em}{\centering{128,64,32,64,64}} & \multirow{2}{5 em}{\centering{\textbf{0.183(3.4\%)}}}\\
& &  \\
\hline
% \multirow{2}{5 em}{\centering{BCVe BCVw}} & \multirow{2}{5 em}{} & \multirow{2}{4.5 em}{\centering{\textbf{0.418(1.5\%)} 0.413}} & \multirow{2}{4.5 em}{\centering{0.181 0.162}}\\ 
% & & & \\
% \hline
\end{tabular}

\caption{Ensemble methods on motivating example with degree labels.}
\label{table:motivationdegree}. 
    \vspace{-1.5em}
\end{table}

The result of the node classification for the degree labels of the merged synthetic graph is shown in Table \ref{table:motivationdegree}. The embedding obtained from the state-of-the-art methods and the ensemble approach is utilized to predict the degree labels. It can be seen that compared to the state-of-the-art algorithms, the ensemble based approach is able to achieve 3.4\% improvement in macro F1 score. Although not significant, it is still able to improve the classification accuracy. 

\begin{table}[!ht]
\centering
 \vspace{-0.5em}
\begin{tabular}{c|c|c}
\hline
   
Methods & Dimensions &  Macro-F1\\
\hline
gf & 64 & 0.108\\
lap & 32 & 0.064\\
hope & 64 & 0.090\\
sdne & 128 & \underline{0.191} \\
node2vec & 128 & 0.142\\ 
\hline
\multirow{2}{7 em}{\centering{sdne, node2vec gf,hope,lap}} & \multirow{2}{7 em}{\centering{128,64,64,32,128}}& \multirow{2}{5 em}{\centering{\textbf{0.215(12.6\%)}}}\\
& & \\
\hline
% \multirow{2}{5 em}{\centering{BCVe BCVw}} & \multirow{2}{5 em}{} & \multirow{2}{4.5 em}{\centering{0.418 \textbf{0.452(4.1\%)} }} & \multirow{2}{4.5 em}{\centering{0.185 0.171}}\\ 
% & & & \\
% \hline
\end{tabular}

\caption{Ensemble methods on motivating example with centrality labels.}
\label{table:motivationcommunity}
    \vspace{-1em}
\end{table}

The classification accuracy results for classifying the centrality measures are shown in Table \ref{table:motivationcommunity}. For this label, it can be observed that the ensemble based method is able to achieve 12.6\% improvement in macro F1-score. Both the macro F1-score proves that the ensemble based approach are able to utilize the best characteristic of different graph embedding algorithm's ability to capture the structure of the network. 
\section{Graph Representation Ensemble Learning}\label{sec:methodology}
% \begin{figure*}[!ht]
% \centering
% \includegraphics[width=0.95\textwidth]{img/architectures.pdf}
% \caption{Proposed architectures for dynamic graph embedding.}
% \end{figure*}

In this section, we define the notations and provide the graph ensemble problem statement. We then explain multiple variations of deep learning models capable of capturing temporal patterns in dynamic graphs.
Finally, we design the loss functions and optimization approach.

\subsection{Notations} We define a directed graph as $G = (V, E)$, where $V$ is the vertex set and E is the directed edge set. The adjacency matrix is denoted as $A$. We define the embedding matrix from a method $m$ as $X^m$. The embedding matrix can be used to reconstruct the distance between all pairwise nodes in the graph. We denote this as $S^m$, in which $S_{i, j}^m = \|X_{i, .}^m - X_{j, .}^m\|$.

\subsection{Problem Statement} In this paper, we introduce the problem of ensemble learning on graph representation learning.
We define it as follows: \textit{Given a set of embedding methods $\lbrace m_1, \ldots, m_k \rbrace$ with corresponding embeddings for a graph $G$ as $\lbrace X^{m_1}, \ldots, X^{m_k} \rbrace$ and errors $\lbrace \epsilon_1, \ldots, \epsilon_k \rbrace$ on a graph task $\mathcal{T}$, a graph ensemble learning approach aims to learn an embedding $X^m$ with error $\epsilon$ such that $\epsilon < min(\epsilon_1, \ldots, \epsilon_k)$.}

\subsection{Measuring Graph Embedding Diversity}
Different graph embedding techniques vary in the types of properties of the graphs preserved by them and the model defined.
Broadly, embedding techniques can be divided into: (i) structure preserving, and (ii) community preserving models, defined as follows:

% \textbf{Community Preserving Models} aim to embed nodes 
\defn{ (Community Preserving Models) It aims to embed nodes with lower distance between them closer in the embedding space.}

\defn{ (Structure Preserving Models) It aims to embed structurally similar nodes closer in the embedding space.}

As ensemble accuracy of a combination of methods depends on the diversity of the input methods~\cite{dietterich2002ensemble}, we now establish bounds on the diversity of embedding models.
Graph embedding of a graph $G$ is a matrix $X \in \mathbb{R}^{n\times d}$ where $n$ is the number of nodes and $d$ is the dimension of the embedding.
Thus, we require a diversity measure which can quantify diversity between matrices.
Pearson correlation~\cite{benesty2009pearson} is a popular metric traditionally used to measure diversity of two uni-variate random variables.
It can be generalized to multivariate case and defined as RV coefficient~\cite{robert1976unifying}.
% \defn{\cite{robert1976unifying} (RV Coefficient): Suppose that $X$ and $Y$ are matrices of centered random vectors (column vectors) with covariance matrix given by $\Sigma_{XY} = E(X^T Y)$, the RV coefficient is defined by:

% \begin{align*}
%     RV(X, Y) = \frac{Tr(\Sigma_{XY}\Sigma_{YX})}{\sqrt{Tr(\Sigma_{XX}^2)Tr(\Sigma_{YY}^2)}}.
% \end{align*}}
As RV Coefficient measures linear dependence between the variables and embedding methods can be non-linear in construction, we can use a distance based metric to capture such non-linearity between embeddings:

\defn{\cite{szekely2009brownian} (Distance Covariance): Suppose that $X$ and $Y$ are matrices of centered random vectors (column vectors). Let the $n\times n$ distance matrices $(a_{j, k})$ and $(b{j, k})$ containing all pairwise distances, $a_{j, k} = \| X_j - X_k \|$ and $b_{j, k} = \| Y_j - Y_k \|$. We compute the doubly centered distance matrices $(A_{j, k})$ and $(B_{j, k})$, where $A_{j, k} = a_{j, k} - a_{j, .} - a_{., k} + a_{., .}$ and $B_{j, k} = b_{j, k} - b_{j, .} - b_{., k} + b_{., .}$. The distance covariance is defined as follows:
\begin{align*}
    dCov^2(X, Y) = \frac{1}{n^2}\sum_{j=1}^n \sum_{k=1}^n A_{j, k}B_{j, k}.
\end{align*}
}

\defn{\cite{szekely2007measuring} (Distance Correlation): The distance correlation between random variables $X$ and $Y$ is given as follows:
\begin{align*}
   dCor(X, Y) = \frac{dCov(X, Y)}{\sqrt{dCov(X, X) dCov(Y, Y)}} 
\end{align*}}

Based on this, we obtain the following bound:

\begin{thm}\label{thm:thm1}
Consider two embedding methods $m_1$ and $m_2$ with corresponding embeddings for a graph $G= (V, E)$ as $X^{m_1}$ and $X^{m_2}$, where $
|V| = n$. Let $G$ have a set $V_1$ of structurally similar nodes with $|V_1| = n_1$ and a set $V_2 = V \setminus V_1$ with nodes in multiple communities. If $m_1$ is a purely structural preserving method and $m_2$ preserves both structural and community properties, then distance correlation between the the embeddings has the following bound:
\begin{align*}
    dCor(X^{m_1}, X^{m_2}) < 1 - \frac{n_1}{n}.
\end{align*}
\end{thm}
\begin{proof}
% Without loss of generality, we can assume that the first $n_1$ nodes of graph $G$ are structurally similar.
Let $S^{m_1}$ and $S^{m_2}$ denote the pairwise distance matrices for methods $m_1$ and $m_2$, and $S'^{m_1}$ and $S'^{m_2}$ denote their doubly centered versions.
We now have,
\begin{align}\label{eq:firstSummation}
    dCov(X^{m_1}, X^{m_2}) = \frac{1}{n^2}\sum_{v, w \in V} S'^{m_1}_{v, w}S'^{m_2}_{v, w}
\end{align}
\begin{align}\label{eq:secSummation}
    dCov(X^{m_1}, X^{m_1}) = \frac{1}{n^2}\sum_{v, w \in V} (S'^{m_1}_{v, w})^2
\end{align}

We can divide the first summation (eqn. \ref{eq:firstSummation}) into four parts:
% \begin{align*}
%     dCov(X^{m_1}, X^{m_2}) &= \frac{1}{n^2}\sum_{j=1}^{n_1} \sum_{k=1}^{n_1} S'^{m_1}_{j, k}S'^{m_2}_{j, k} + \frac{1}{n^2}\sum_{j=1}^{n_1} \sum_{k=n_1+1}^{n} S'^{m_1}_{j, k}S'^{m_2}_{j, k}\\
%     &+ \frac{1}{n^2}\sum_{j=n_1}^{n} \sum_{k=1}^{n_1} S'^{m_1}_{j, k}S'^{m_2}_{j, k} + \frac{1}{n^2}\sum_{j=n_1}^{n} \sum_{k=n_1+1}^{n} S'^{m_1}_{j, k}S'^{m_2}_{j, k}
% \end{align*}
\begin{align*}
    dCov(X^{m_1}, X^{m_2}) &= \frac{1}{n^2}\sum_{v, w \in V_1} S'^{m_1}_{v, w}S'^{m_2}_{v, w} + \frac{1}{n^2}\sum_{v \in V_1, w \in V_2} S'^{m_1}_{v, w}S'^{m_2}_{v, w}\\
    &+ \frac{1}{n^2}\sum_{v \in V_2, w \in V_1} S'^{m_1}_{v, w}S'^{m_2}_{v, w} + \frac{1}{n^2}\sum_{v, w \in V_2} S'^{m_1}_{v, w}S'^{m_2}_{v, w}
\end{align*}

As $m_2$ preserves structural similarity, the distance between each pair of nodes in set $V_1$ will be 0 yielding the first term of above equation 0. Also, since $V_1$ and $V_2$ do not have a specified relation, the embedding distances by $m_1$ and $m_2$ will be randomly distributed and uncorrelated. Thus, the second and third terms become 0.
We can get similar results for second summation (eqn. \ref{eq:secSummation}) as well. From this, we get
\begin{align*}
    dCor(X^{m_1}, X^{m_2}) = \frac{\sum_{v, w \in V_2} S'^{m_1}_{v, w}S'^{m_2}_{v, w}}{\sqrt{\sum_{v, w \in V_2} (S'^{m_1}_{v, w})^2}\sqrt{\sum_{v, w \in V_2} (S'^{m_2}_{v, w})^2}}
\end{align*}

As correlation between two variables is bounded by 1, from the above we get
\begin{align*}
    dCor(X^{m_1}, X^{m_2}) &\leq \frac{(n - n_1)^2}{n^2} = (1 - \frac{n_1}{n})^2 = 1 + \frac{n_1^2}{n^2} - \frac{2n_1}{n}
\end{align*}
Also, $n_1 < n$ and thus $\frac{n_1^2}{n^2} < \frac{n_1}{n}$. We thus get
\begin{align*}
    dCor(X^{m_1}, X^{m_2}) < 1 - \frac{n_1}{n}.
\end{align*}
\end{proof}

\begin{cor}
For a graph $G$ with $s$ sets of structurally similar nodes $\lbrace V_1 \ldots V_k\rbrace$ with $|V_i| = n_i$ and embedding methods $m_1$ and $m_2$ preserving purely structural and both structural and community properties respectively, the distance correlation bound is:
\begin{align*}
    dCor(X^{m_1}, X^{m_2}) < 1 - \sum_{i=1}^s\frac{n_i}{n}.
\end{align*}
\end{cor}
\begin{proof}
The summation in Theorem~\ref{thm:thm1}, eqn.~\ref{eq:secSummation}, can be broken down into $2s$ parts and the rest follows as above.
\end{proof}

% The algorithm is specified in Algorithm \ref{alg:elaine}.

\subsection{Measuring Label Prediction Diversity}
We have now established the upper bound on correlation between the embeddings. We also know the following about predictions using Logistic Regression:

\begin{thm}\label{thm:thm2}
Consider two sets of feature spaces for data $\mathcal{D}$ represented as $X \in \mathbb{R}^{n\times d_1}$ and $Y \in \mathbb{R}^{n\times d_2}$ with labels for individual data points as $Z \in \mathbb{R}^n$. If logistic regression models trained on $(X, Z)$ and $(Y, Z)$ obtain accuracy of $a_X$ and $a_Y$ respectively, then we have the following bound for the model trained on $(X\mathbin\Vert Y, Z)$, where $\mathbin\Vert$ denotes concatenation operation:
\begin{align*}
    a_{X\mathbin\Vert Y} \geq \max(a_X, a_Y)
\end{align*}
\end{thm}
\begin{proof}
Without loss of generality, assume that $a_X > a_Y$. As logistic regression is an additive model, setting weights of the model corresponding to $Y$ would yield the accuracy of the concatenated model $a_X$.
\end{proof}

From the above theorem, we note that adding embeddings of method $m_2$ on $m_1$ would not decrease the performance.
Further, the equality in Theorem~\ref{thm:thm2} is realized when $Y$ is a linear scaling of $X$ or distances in $Y$ are exactly correlated with $X$.
But from Theorem~\ref{thm:thm1} we have an upper bound on the correlation between the embeddings.
Thus, we can get $a_{X\mathbin\Vert Y} > \max(a_X, a_Y)$.
Tighter bounds are left as a future work.

\subsection{Runtime Optimization Techniques}
Given a set of $k$ embedding methods $\lbrace m_1 \ldots m_k \rbrace$ with optimal hyperparameters $\lbrace \mathcal{h}_1 \ldots \mathcal{h}_k\rbrace$ and the maximum time complexity from the methods as $T$ per unit dimension, a naive implementation of finding the optimal combination of methods would take a time complexity of $O(2^k \times T \times d)$, where $d$ is the embedding dimensionality. To optimize this, we do an approximation by greedily adding the next method's embedding to the current set of embeddings. This yields a time complexity of $O(k\times T \times d)$.

\subsection{Algorithm}
\begin{algorithm}[tb]\color{black}
\Fn{graphensemble (Graph $G$, Embedding methods $\mathcal{M} = \lbrace m_1, .., m_K\rbrace$)}{
train, val, test $\leftarrow$ splitNodeIndexes($G$)\;
 \For{$i=1 \ldots K$}{
     $X^{m_i}$ $\leftarrow$ getEmbedding($G$, $m_i$)\;
     training(model, $X^{m_i}$[train])\;
    %  X^{m_i}_{train}, X^{m_i}_{val}, X^{m_i}_{test}  = splitdata(X^{m_i})\;
     $a_i$ $\leftarrow$ evaluate(model, $X^{m_i}$[val])\;
 }
 sortedmethods $\leftarrow$ Sort $\mathcal{M}$ based on $a$\;
 
 ensembleset $\leftarrow$ set()\;
 $a_{prev}$ $\leftarrow$ 0\;
  \For{$m$ $\in$ sortedmethods }{
     ensembleset.add(m)\;
     $X$ $\leftarrow$ concat($X^{m_i}$ $\forall$ $m_i$ $\in$ ensembleset)\;
     training(model, $X$[train])\;
     $a_e$ $\leftarrow$ evaluate(model, $X$[val])\;
     \If{$a_e$ $<$ $a_{prev}$}{
         ensembleset.remove(m)\;
     }
     $a_{prev}$ $\leftarrow$ $a_e$\;
    
 }
 $X$ $\leftarrow$ concat($X^{m_i}$ $\forall$ $m_i$ $\in$ ensembleset)\;
 training(model, $X$[train])\;
 $a_{test}$ $\leftarrow$ evaluate(model, $X$[test])\;
 return $a_{test}$}
\caption{graphensemble}
\label{alg:ensemble}
\end{algorithm}
Algorithm~\ref{alg:ensemble} provides the pseudo-code for the framework.
Given an input graph $G$, we split the graph nodes into training, validation and test.
We then use the validation set to get an accuracy score for each embedding method.
Based on this, we greedily add the next best embedding approach to evaluate the performance of the ensemble of methods.
Finally, we report the performance on a held-out test set.
In the experiments below the above step is performed 5 times and the average is reported.
% \subsection{Time Complexity}
% \subsection{Timing Analysis}
\section{Experiments}\label{sec:ens_exp}
In this section, we establish the Graph Ensemble approach against five state-of-the-art baseline embedding methods to evaluate their multi-label node classification performance on four benchmark datasets. In addition, we yield insights into the correlation of graph embedding obtained by the different methods.

\subsection{Datasets}
\begin{table}[!t]
\label{balance}
\centering
\begin{tabular}{c|c|c|c}
\hline
Dataset & Nodes  & Edges & Classes\\
\hline
PPI & 3,890  & 38,839 & 50\\
BlogCatalog & 10,312 & 333,983 & 39\\
Citeseer & 3,312 & 4,660 & 6\\
Wikipedia & 4,777 & 92,512 & 40 \\
% \multirow{2}{5.5 em}{\centering{node2vec,sdne hope,lap,gf}} & \multirow{2}{5 em}{\centering{64,128 64,64,32}} & \multirow{2}{4.5 em}{\centering{\textbf{0.402 (20.7\%)}}} & \multirow{2}{4.5 em}{\centering{\textbf{0.201 (16.9\%)}}}\\
% & & & \\
\hline
\end{tabular}
\label{tab:data}
\caption{Statistics of benchmark datasets in the experiment}
\end{table}
As Table 3 shows, we use four benchmark real-life graphs for node classification task in our experiment. For each dataset, we derive the largest weakly connected component from the original graph.

\begin{itemize}
	\item Protein-Protein Interactions (PPI)\cite{breitkreutz2008biogrid}: This is a network of biological interactions between proteins in humans. This network has 3,890 nodes and 38,739 edges.
	\item BlogCatalog\cite{tang2009relational}:  This is a network of social relationships of the bloggers listed on the BlogCatalog website. The labels represent blogger interests inferred through the metadata provided by the bloggers. The network has 10,312 nodes, 333,983 edges and 39 different labels.
	\item Citeseer\cite{lu2003link}: This dataset consists of 3312 scientific publications classified into one of six classes. The citation network consists of 4732 links.
	\item Wikipedia\cite{mahoney2011large}: This is a cooccurrence network of words appearing in the first million bytes of the Wikipedia dump. The labels represent the Part-of-Speech (POS) tags inferred using the Stanford POS-Tagger. The network has 4,777 nodes, 184,812 edges, and 40 different labels.
\end{itemize}

\subsection{Baseline Graph Embedding Methods}
% Different methods aim to capture different properties of the original graph and yield to variant embeddings. Matrix factorization based models, such as GF and LAP, decompose graph matrix to employ spectral learning. Autoencoder based approach like SDNE captures first and second network proximities. HOPE preserves higher order proximity of nodes. Random walk based method Node2vec preserves both graph community structure and node structure equivalence.
We compare our Graph Ensemble method with the following five baseline graph embedding models. 

\begin{itemize}
	\item Graph Factorization (GF)\cite{ahmed2013distributed}: It factorizes the adjacency matrix with regularization.
	\item Laplacian Eigenmaps (LAP)\cite{belkin2002laplacian}: It preserves local information by projecting points into a low-dimensional space using eigen-vectors of the graph.
	\item High Order Proximity Preserving(HOPE)\cite{ou2016asymmetric}: It factorizes the higher order similarity matrix between nodes using generalized SVD.
	\item Structural Deep Network Embedding(SDNE)\cite{wang2016structural}: This uses deep auto-encoders to preserve the first and second order network proximities by using non-linear functions to obtain the embedding. 
	\item node2vec\cite{grover2016node2vec}: It is an embedding technique that uses random walks on graphs to obtain node representations which preserves higher order proximity between nodes.
\end{itemize}

\subsection{Graph Ensemble Approach}
Our graph representation ensemble learning mechanism leverages a bag of single embedding methods and achieves an optimal embedding combination for graph feature learning. First, we run single graph embedding methods on the original graph to get the best embedding at each dimension. Then, we use the greedy approximated search to add embedding generated by other methods iteratively to the embedding given by the best single method. In the end, we feed the ensemble concatenation embedding and baseline method embedding to the downstream multi-label node classification task. At each experiment round, we split the nodes of a graph into training data (50\%), validation data (20\%) and test data (30\%). Using training data is intended to find the best hyperparamter for single methods. We choose the optimal ensemble embedding combination based on the validation data. And we report the performance of our graph ensemble methods and five baseline methods on test data.

\subsubsection{Hyperparameter Search}
In order to get the best embedding for each single graph embedding model, we employ a best hyperparamter search on the training dataset. Among three embedding dimensions 32, 64 and 128, we select the best hyperparameter set respectively at each dimension. Except for LAP which does not contain hyperparamters, we use grid search on a range of hyperparameter sets for the other four methods. For GF, we search parameters including learning rate from \{1e-3, 1e-2, 1e-1\} and regularization from \{1e-1, 1, 10\}. For HOPE, we select a decaying factor from \{1e-4, 1e-3, 1e-2, 1e-1\} and similarity function from Katz Index, PageRank, Common Neighbours and Adamic-Adar. For SDNE, we fix the autoencoder structure 500, 1000, 300 nodes in each layer, and set first loss function parameter $\alpha$ to 1e-5 and penalty $\beta$ to 10. We select two regularization factors and $xeta$ from \{1e-3, 1e-2\} respectively. As for Node2vec, we set walk length to 80, number of walks to 10, context size to 10. We select return $p$ and in-and-out $q$ from \{0.25, 0.5, 1, 2, 4\} respectively.

\subsubsection{Ensemble Combination Search}
After obtaining the best hyperparameter set for each method at each dimension, we evaluate their performance on multi-label node classification task with validation dataset and select the optimal ensemble combination. First, we choose the best method which has best performance on the training data. We test its performance on validation data under best setting in respect to three dimensions 32, 64 and 128, and then select its best dimension based on Macro $F_1$ score. Secondly, we append the embedding of the second best method at three dimensions separately to the best embedding so far and repeat the evaluation process. If the performance improves, we keep the second embedding at the chosen dimension. Otherwise we abandon this method and continue the appending process. In the end, we will obtain the best combination iteratively via such greedy approximation.

\subsection{Embedding Correlation }
% \begin{figure*}[!ht]
% \centering
% \vspace{-0.8em}
% \includegraphics[width=0.9\textwidth]{img/rv.PNG}
% \vspace{-1em}
% \caption{RV Coefficient correlations of embedding methods on real networks (dimensions set to 128).}
% % \vspace{+0.8em}
% \label{fig:rv}
% \end{figure*}
\begin{figure*}[!ht]
\centering
% \vspace{-0.8em}
\includegraphics[width=0.9\textwidth]{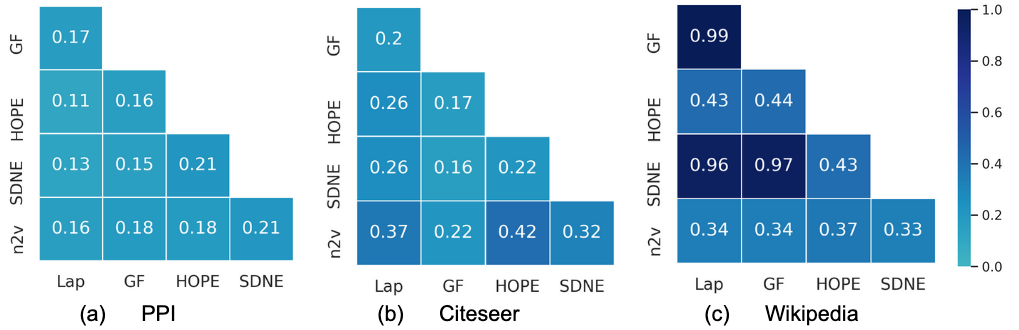}
\vspace{-1em}
\caption{Distance correlations of embedding methods on real networks (dimensions set to 128).}
\vspace{-0.8em}
\label{fig:dc}
\end{figure*}
The distance correlations between the embeddings obtained by different embedding methods is presented in Figure~\ref{fig:dc}.
We observe that the correlation between the embeddings varies significantly with the underlying data set.
For PPI and Citeseer, we see that all methods are weakly correlated.
This strengthens our claim in Theorem~\ref{thm:thm1} that embedding methods preserve different properties and if the underlying graph is complex, then the embeddings will be diverse.
For the Wikipedia dataset we observe that Graph Factorization and Laplacian Eigenmaps have a very high correlation.
As they both capture first order proximity, the correlation may be because the first order correlations in the Wikipedia dataset may have a simple pattern easily visible to both these approaches.
We also observe that SDNE which preserves first and second order proximity in a non-linear way also has high correlations with GF and Lap further strengthening our claim.

\subsection{Multi-label Node Classification}
In the multi-label node classification task, we are given a graph as well as labels of a proportion of nodes as training data. And we aim to predict the unknown labels for the rest of nodes in the test data. Each node in the graph has one or multiple labels. To evaluate the graph ensemble embedding and baseline methods embedding, we utilize the same One-Vs-the-Rest multi-label strategy and Logistic Regression by default setting to build classifiers. To ensure the robustness of our proposed graph ensemble methods and stability of the experiments, we repeat the whole process for 5 rounds and report the average results. We use Macro $F_1$ and Micro $F_1$ as evaluation metrics. Micro $F_1$ has similar performance like Macro $F_1$ thus it is not reported in the paper. We care more about the minority class prediction and Macro $F_1$ is preferably considered. 

We summarize multi-label classification results in Table~\ref{tab: ensembleresults}.
Overall, we observe that the ensemble of methods outperforms individual methods significantly with the exception of Citeseer.
$node2vec$ gives highest accuracy for all data sets except Wikipedia for which HOPE outperforms other methods.
Another key observation is that the optimal embedding dimensionality for a method in an ensemble may be different than the individual optimal.
This can be attributed to the interplay of embeddings when concatenated together and the amount of information shared between them.

\subsection{Minority Class}

\begin{figure}[h!]%
    \centering
    \vspace{-1em}
    \subfloat[Node Classification on Citeseer]{{\includegraphics[width=0.24\textwidth]{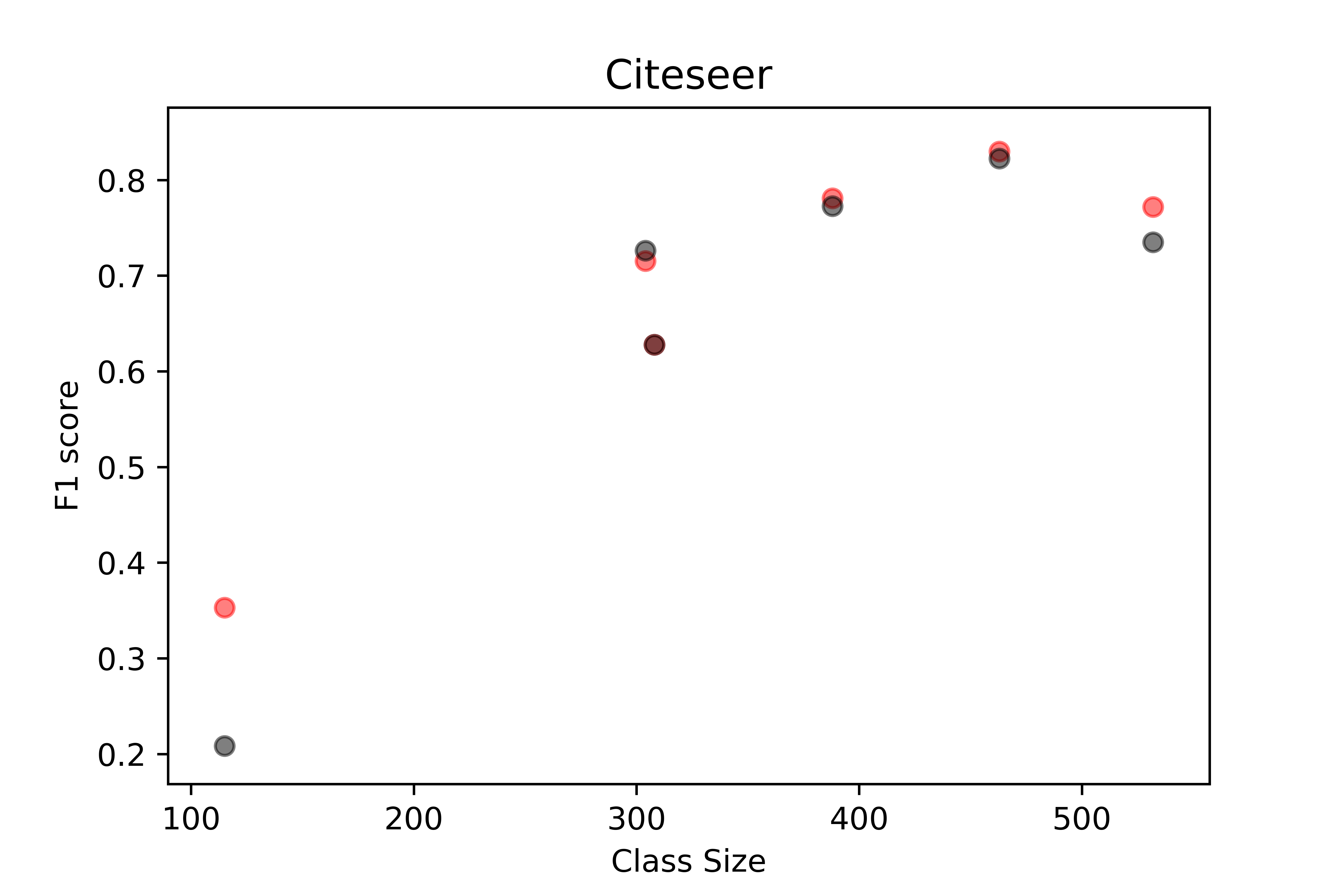}}}%
    % \qquad
    \subfloat[Node Classification on Wikipedia]{{\includegraphics[width=0.24\textwidth]{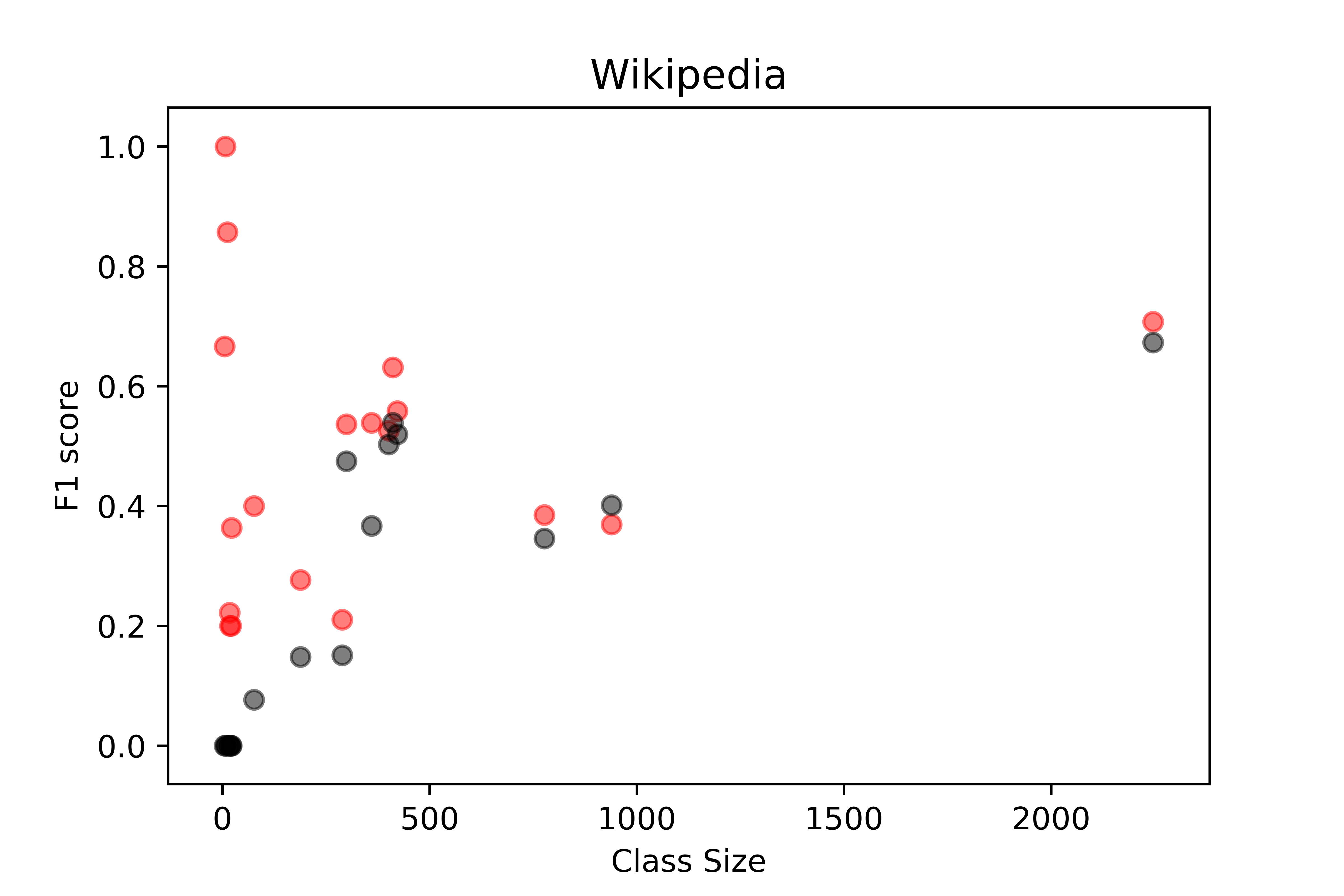}}}%
    \caption{Node classification results on two Citeseer and Wikipedia Dataset. Y-axis is $F_1$ score on each class. The red spots represent the graph ensemble learning and the black spots represent the best single graph embedding method's performance}%
    \vspace{-1em}
    \label{fig:motiv2}%
\end{figure}

As the Figure 5 indicates, the $F_1$ score of our graph ensemble methods on smaller classes are higher than the best individual methods. Our graph ensemble strategy combines the captured features derived by all single methods and generate a comprehensive graph embedding, which is able to improve the performance on less represented classes.
In Wikipedia, we observe that for really small classes, none of the individual methods perform well and give close to 0 F1. However, the combination ensemble is able to perform well and gives F1 ranging from 0.2 to 1.0. Similarly in Citeseer, we see an improvement of about 50\% for less represented labels.

\begin{table}[!ht]\footnotesize
\label{balance}
\centering
\begin{tabular}{c|c|c|c}
\hline
Dataset & Method & Dimensions & Macro-F1\\
\hline
\multirow{6}{4.5 em}{\centering{PPI}} & gf & 128 & 0.118 \\ & lap & 128 & 0.077 \\& hope & 128 & 0.144 \\& sdne & 32 & 0.159 \\& node2vec & 128 & \underline{0.179}\\
 & node2vec,hope,gf,lap & 64,32,128,64 & \textbf{0.192 (7.3\%)}\\
\hline

\multirow{6}{4.5 em}{\centering{BlogCatalog}}& gf & 128 & 0.044\\ &lap & 128 & 0.047 \\& hope & 128 & 0.137 \\& sdne & 128 & 0.212\\& node2vec & 128 & \underline{0.225}\\
 & node2vec,sdne,lap,gf & 128,32,128,128 & \textbf{0.243 (8.0\%)}\\
\hline

\multirow{6}{4.5 em}{\centering{Citeseer}}& gf & 128 & 0.442\\ & lap & 32 & 0.388 \\ & hope & 128 & 0.517\\& sdne & 128 & 0.513\\ &  node2vec & 128 & \underline{0.671}\\
 & node2vec,sdne,hope,gf,lap & 128,128,128,32,32 & \textbf{0.673 (0.3\%)}\\
\hline

\multirow{6}{4.5 em}{\centering{Wikipedia}}& gf & 64 & 0.042\\ & lap & 128 & 0.034\\ & hope & 128 & \underline{0.172} \\ & sdne & 128 & 0.032 \\& node2vec & 64 & 0.110\\
 &hope,node2vec,sdne,gf,lap & 128,64,128,128,64 & \textbf{0.181 (5.2\%)}\\
\hline
% lap & 32 & 0.197 & 0.037\\
% hope & 128 &\textbf{ 0.333} & \textbf{0.172}\\
% sdne & 128 & 0.316 & 0.155 \\
% node2vec & 32 & 0.274 & 0.120\\ 
% \hline
% \multirow{2}{5.5 em}{\centering{node2vec,sdne hope,lap,gf}} & \multirow{2}{5 em}{\centering{64,128 64,64,32}} & \multirow{2}{4.5 em}{\centering{\textbf{0.402 (20.7\%)}}} & \multirow{2}{4.5 em}{\centering{\textbf{0.201 (16.9\%)}}}\\
% & & & \\
% \hline
\end{tabular}
\label{tab: ensembleresults}
\caption{Macro $F_1$ score of Graph Ensemble methods and five baseline graph embedding methods. The scores with underline show the best performance of single method on different datasets. The percentage inside parentheses indicates the performance gain of Graph Ensemble method over the best single method. }
\end{table}

\section{Conclusion}\label{sec:conclusion}
% This paper introduced dyngraph2vec, a model for capturing temporal patterns in dynamic networks.
% It learns the evolution patterns of individual nodes and provides an embedding capable of predicting future links with higher precision.
% We propose three variations of our model based on the architecture with varying capabilities.
% The experiments show that our model can capture temporal patterns on synthetic and real datasets and outperform state-of-the-art methods in link prediction. There are several directions for future work: (1) interpretability by extending the model to provide more insight into network dynamics and better understand temporal dynamics; (2) automatic hyperparameter optimization for higher accuracy; and (3) graph convolutions to learn from node attributes and reduce the number of parameters. 

In this paper, we proposed a Graph Representation Ensemble Learning framework which can create an ensemble of graph embedding approaches outperforming each individual method.
We provided theoretical analysis of the framework and established the upper bound on the correlations between graph embedding techniques.
Further, we compared our method with state-of-the-art embedding methods and showed improvement on four real world networks.
We also showed that the model is even more useful for underrepresented classes.
There are several research directions for future work: (1) tighter ensemble bound to get a better understanding of the framework, (2) information theoretic approaches which can take into account the mutual information between embeddings, and (3) dynamic ensembles which can create ensemble learning for evolving graphs.

% \section{Acknowledgements}

\bibliographystyle{aaai}
\bibliography{bibliography}

\begin{thebibliography}{}

\bibitem[\protect\citeauthoryear{Ahmed \bgroup et al\mbox.\egroup
  }{2013a}]{Ahmed2013}
Ahmed, A.; Shervashidze, N.; Narayanamurthy, S.; Josifovski, V.; and Smola,
  A.~J.
\newblock 2013a.
\newblock Distributed large-scale natural graph factorization.
\newblock In {\em Proceedings of the 22nd international conference on World
  Wide Web},  37--48.
\newblock ACM.

\bibitem[\protect\citeauthoryear{Ahmed \bgroup et al\mbox.\egroup
  }{2013b}]{ahmed2013distributed}
Ahmed, A.; Shervashidze, N.; Narayanamurthy, S.; Josifovski, V.; and Smola,
  A.~J.
\newblock 2013b.
\newblock Distributed large-scale natural graph factorization.
\newblock In {\em Proceedings of the 22nd international conference on World
  Wide Web},  37--48.
\newblock ACM.

\bibitem[\protect\citeauthoryear{Belkin and Niyogi}{2001}]{belkin2001laplacian}
Belkin, M., and Niyogi, P.
\newblock 2001.
\newblock Laplacian eigenmaps and spectral techniques for embedding and
  clustering.
\newblock In {\em NIPS}, volume~14,  585--591.

\bibitem[\protect\citeauthoryear{Belkin and Niyogi}{2002}]{belkin2002laplacian}
Belkin, M., and Niyogi, P.
\newblock 2002.
\newblock Laplacian eigenmaps and spectral techniques for embedding and
  clustering.
\newblock In {\em Advances in neural information processing systems},
  585--591.

\bibitem[\protect\citeauthoryear{Benesty \bgroup et al\mbox.\egroup
  }{2009}]{benesty2009pearson}
Benesty, J.; Chen, J.; Huang, Y.; and Cohen, I.
\newblock 2009.
\newblock Pearson correlation coefficient.
\newblock In {\em Noise reduction in speech processing}. Springer.
\newblock  1--4.

\bibitem[\protect\citeauthoryear{Breitkreutz \bgroup et al\mbox.\egroup
  }{2008}]{breitkreutz2008biogrid}
Breitkreutz, B.-J.; Stark, C.; Reguly, T.; Boucher, L.; Breitkreutz, A.;
  Livstone, M.; Oughtred, R.; Lackner, D.~H.; B{\"a}hler, J.; Wood, V.; et~al.
\newblock 2008.
\newblock The biogrid interaction database: 2008 update.
\newblock {\em Nucleic acids research} 36(suppl 1):D637--D640.

\bibitem[\protect\citeauthoryear{Bruna \bgroup et al\mbox.\egroup
  }{2013}]{bruna2013spectral}
Bruna, J.; Zaremba, W.; Szlam, A.; and LeCun, Y.
\newblock 2013.
\newblock Spectral networks and locally connected networks on graphs.
\newblock {\em arXiv preprint arXiv:1312.6203}.

\bibitem[\protect\citeauthoryear{Cao, Lu, and Xu}{2016}]{cao2016deep}
Cao, S.; Lu, W.; and Xu, Q.
\newblock 2016.
\newblock Deep neural networks for learning graph representations.
\newblock In {\em Proceedings of the Thirtieth AAAI Conference on Artificial
  Intelligence},  1145--1152.
\newblock AAAI Press.

\bibitem[\protect\citeauthoryear{Dietterich and
  others}{2002}]{dietterich2002ensemble}
Dietterich, T.~G., et~al.
\newblock 2002.
\newblock Ensemble learning.
\newblock {\em The handbook of brain theory and neural networks} 2:110--125.

\bibitem[\protect\citeauthoryear{Freeman}{1978}]{freeman1978centrality}
Freeman, L.~C.
\newblock 1978.
\newblock Centrality in social networks conceptual clarification.
\newblock {\em Social networks} 1(3):215--239.

\bibitem[\protect\citeauthoryear{Freeman}{2000}]{freeman2000visualizing}
Freeman, L.~C.
\newblock 2000.
\newblock Visualizing social networks.
\newblock {\em Journal of social structure} 1(1):4.

\bibitem[\protect\citeauthoryear{Gehrke, Ginsparg, and
  Kleinberg}{2003}]{Gehrke2003}
Gehrke, J.; Ginsparg, P.; and Kleinberg, J.
\newblock 2003.
\newblock Overview of the 2003 kdd cup.
\newblock {\em ACM SIGKDD Explorations} 5(2).

\bibitem[\protect\citeauthoryear{Goyal and Ferrara}{2018}]{goyal2017graph}
Goyal, P., and Ferrara, E.
\newblock 2018.
\newblock Graph embedding techniques, applications, and performance: A survey.
\newblock {\em Knowledge-Based Systems}.

\bibitem[\protect\citeauthoryear{Goyal, Sapienza, and
  Ferrara}{2018}]{goyal2018recommending}
Goyal, P.; Sapienza, A.; and Ferrara, E.
\newblock 2018.
\newblock Recommending teammates with deep neural networks.
\newblock In {\em Proceedings of the 29th on Hypertext and Social Media},
  57--61.
\newblock ACM.

\bibitem[\protect\citeauthoryear{Grover and Leskovec}{2016a}]{Grover2016}
Grover, A., and Leskovec, J.
\newblock 2016a.
\newblock node2vec: Scalable feature learning for networks.
\newblock In {\em Proceedings of the 22nd International Conference on Knowledge
  Discovery and Data Mining},  855--864.
\newblock ACM.

\bibitem[\protect\citeauthoryear{Grover and
  Leskovec}{2016b}]{grover2016node2vec}
Grover, A., and Leskovec, J.
\newblock 2016b.
\newblock node2vec: Scalable feature learning for networks.
\newblock In {\em Proceedings of the 22nd ACM SIGKDD international conference
  on Knowledge discovery and data mining},  855--864.
\newblock ACM.

\bibitem[\protect\citeauthoryear{Hansen and Salamon}{1990}]{hansen1990neural}
Hansen, L.~K., and Salamon, P.
\newblock 1990.
\newblock Neural network ensembles.
\newblock {\em IEEE Transactions on Pattern Analysis \& Machine Intelligence}
  (10):993--1001.

\bibitem[\protect\citeauthoryear{Henaff, Bruna, and
  LeCun}{2015}]{henaff2015deep}
Henaff, M.; Bruna, J.; and LeCun, Y.
\newblock 2015.
\newblock Deep convolutional networks on graph-structured data.
\newblock {\em arXiv preprint arXiv:1506.05163}.

\bibitem[\protect\citeauthoryear{Kipf and Welling}{2016a}]{kipf2016semi}
Kipf, T.~N., and Welling, M.
\newblock 2016a.
\newblock Semi-supervised classification with graph convolutional networks.
\newblock {\em arXiv preprint arXiv:1609.02907}.

\bibitem[\protect\citeauthoryear{Kipf and Welling}{2016b}]{kipf2016variational}
Kipf, T.~N., and Welling, M.
\newblock 2016b.
\newblock Variational graph auto-encoders.
\newblock {\em arXiv preprint arXiv:1611.07308}.

\bibitem[\protect\citeauthoryear{Liaw, Wiener, and
  others}{2002}]{liaw2002classification}
Liaw, A.; Wiener, M.; et~al.
\newblock 2002.
\newblock Classification and regression by randomforest.
\newblock {\em R news} 2(3):18--22.

\bibitem[\protect\citeauthoryear{Lu and Getoor}{2003}]{lu2003link}
Lu, Q., and Getoor, L.
\newblock 2003.
\newblock Link-based classification.
\newblock In {\em ICML}, volume~3,  496--503.

\bibitem[\protect\citeauthoryear{Mahoney}{2011}]{mahoney2011large}
Mahoney, M.
\newblock 2011.
\newblock Large text compression benchmark.
\newblock {\em URL: http://www. mattmahoney. net/text/text. html}.

\bibitem[\protect\citeauthoryear{Opitz}{1999}]{opitz1999feature}
Opitz, D.~W.
\newblock 1999.
\newblock Feature selection for ensembles.
\newblock {\em AAAI/IAAI} 379:384.

\bibitem[\protect\citeauthoryear{Ou \bgroup et al\mbox.\egroup
  }{2016a}]{Ou2016}
Ou, M.; Cui, P.; Pei, J.; Zhang, Z.; and Zhu, W.
\newblock 2016a.
\newblock Asymmetric transitivity preserving graph embedding.
\newblock In {\em Proc. of ACM SIGKDD},  1105--1114.

\bibitem[\protect\citeauthoryear{Ou \bgroup et al\mbox.\egroup
  }{2016b}]{ou2016asymmetric}
Ou, M.; Cui, P.; Pei, J.; Zhang, Z.; and Zhu, W.
\newblock 2016b.
\newblock Asymmetric transitivity preserving graph embedding.
\newblock In {\em Proceedings of the 22nd ACM SIGKDD international conference
  on Knowledge discovery and data mining},  1105--1114.
\newblock ACM.

\bibitem[\protect\citeauthoryear{R{\"a}tsch, Onoda, and
  M{\"u}ller}{2001}]{ratsch2001soft}
R{\"a}tsch, G.; Onoda, T.; and M{\"u}ller, K.-R.
\newblock 2001.
\newblock Soft margins for adaboost.
\newblock {\em Machine learning} 42(3):287--320.

\bibitem[\protect\citeauthoryear{Ricci and Aha}{1997}]{ricci1997extending}
Ricci, F., and Aha, D.~W.
\newblock 1997.
\newblock Extending local learners with error-correcting output codes.
\newblock {\em Naval Center for Applied Research in Artificial Intelligence,
  Washington, DC}.

\bibitem[\protect\citeauthoryear{Robert and
  Escoufier}{1976}]{robert1976unifying}
Robert, P., and Escoufier, Y.
\newblock 1976.
\newblock A unifying tool for linear multivariate statistical methods: the
  rv-coefficient.
\newblock {\em Journal of the Royal Statistical Society: Series C (Applied
  Statistics)} 25(3):257--265.

\bibitem[\protect\citeauthoryear{Sz{\'e}kely \bgroup et al\mbox.\egroup
  }{2007}]{szekely2007measuring}
Sz{\'e}kely, G.~J.; Rizzo, M.~L.; Bakirov, N.~K.; et~al.
\newblock 2007.
\newblock Measuring and testing dependence by correlation of distances.
\newblock {\em The annals of statistics} 35(6):2769--2794.

\bibitem[\protect\citeauthoryear{Sz{\'e}kely, Rizzo, and
  others}{2009}]{szekely2009brownian}
Sz{\'e}kely, G.~J.; Rizzo, M.~L.; et~al.
\newblock 2009.
\newblock Brownian distance covariance.
\newblock {\em The annals of applied statistics} 3(4):1236--1265.

\bibitem[\protect\citeauthoryear{Tang and Liu}{2009}]{tang2009relational}
Tang, L., and Liu, H.
\newblock 2009.
\newblock Relational learning via latent social dimensions.
\newblock In {\em Proceedings of the 15th international conference on Knowledge
  discovery and data mining},  817--826.
\newblock ACM.

\bibitem[\protect\citeauthoryear{Theocharidis \bgroup et al\mbox.\egroup
  }{2009}]{theocharidis2009network}
Theocharidis, A.; Van~Dongen, S.; Enright, A.; and Freeman, T.
\newblock 2009.
\newblock Network visualization and analysis of gene expression data using
  biolayout express3d.
\newblock {\em Nature protocols} 4:1535--1550.

\bibitem[\protect\citeauthoryear{Wang, Cui, and Zhu}{2016a}]{Wang2016}
Wang, D.; Cui, P.; and Zhu, W.
\newblock 2016a.
\newblock Structural deep network embedding.
\newblock In {\em Proceedings of the 22nd International Conference on Knowledge
  Discovery and Data Mining},  1225--1234.
\newblock ACM.

\bibitem[\protect\citeauthoryear{Wang, Cui, and
  Zhu}{2016b}]{wang2016structural}
Wang, D.; Cui, P.; and Zhu, W.
\newblock 2016b.
\newblock Structural deep network embedding.
\newblock In {\em Proceedings of the 22nd ACM SIGKDD international conference
  on Knowledge discovery and data mining},  1225--1234.
\newblock ACM.

\bibitem[\protect\citeauthoryear{Zhou}{2012}]{zhou2012ensemble}
Zhou, Z.-H.
\newblock 2012.
\newblock {\em Ensemble methods: foundations and algorithms}.
\newblock Chapman and Hall/CRC.

\end{thebibliography}
\end{document}